\documentclass{amsart}

\usepackage{amssymb,amsmath}
\usepackage{verbatim}

\usepackage{comment}
\usepackage{cancel}

\usepackage{tikz}
\usepackage{pgfplots}

\usepackage[section]{algorithm}
\usepackage{algpseudocode}

\usepackage{cleveref}
\usepackage{url}
\usepackage{bm}

\newtheorem{theorem}{Theorem}[section]
\newtheorem{proposition}[theorem]{Proposition}
\newtheorem{lemma}[theorem]{Lemma}
\newtheorem{corollary}[theorem]{Corollary}

\newtheorem{definition}[theorem]{Definition}

\newtheorem{remark}[theorem]{Remark}

\newtheorem{example}[theorem]{Example}

\def\NN{{\mathbb N}}

\def\FF{{\mathbb F}}

\def\Ker{{\mathrm{Ker\,}}}

\def\Gr{Gr\"obner}

\def\GF{{\mathrm{GF}}}

\def\wt{{\mathrm{wt}}}
\def\mod{{\mathrm{mod}}}

\def\GB{{\textsc{Groebner}}}
\def\GBSafe{{\textsc{GroebnerSafe}}}

\def\MultiSolve{{\textsc{MultiSolve}}}
\def\GBDecode{{\textsc{GBDecode}}}

\def\maxdeg{{\mathrm{maxdeg}}}
\def\Oracle{{\textsc{Oracle}}}

\def\OracleH{{\textsc{OracleH}}}
\def\OracleT{{\textsc{OracleT}}}

\def\tame{{\tt tame}}
\def\wild{{\tt wild}}
\def\status{{\tt status}}

\def\cF{{\mathcal{F}}}
\def\cO{{\mathrm{O}}}

\def\cC{{\mathcal{C}}}

\def\bt{{\bar{t}}}

\def\max{{\mathrm{max}}}

\def\Magma{{\sc Magma}}

\begin{document}

\title[Hamming Ideals and Gr\"obner Bases for Syndrome Decoding]
{Hamming Ideals and Gr\"obner Bases for ISD-like Syndrome Decoding}


\author[R. La Scala]{Roberto La Scala$^*$}

\author[M. Marchesin]{Marco Marchesin$^{**}$}

\author[S.K. Tiwari]{Sharwan K. Tiwari$^{\dagger}$}

\address{$^*$ Dipartimento di Fisica, Universit\`a degli Studi di Bari
``Aldo Moro'', Via Orabona 4, 70125 Bari, Italy}
\email{roberto.lascala@uniba.it}

\address{$^{**}$ Dipartimento di Matematica, Universit\`a degli Studi di Bari
``Aldo Moro'', Via Orabona 4, 70125 Bari, Italy}
\email{marco.marchesin@uniba.it}

\address{$^{\dagger}$ Cryptography Research Centre, Technology Innovation Institute,
Abu Dhabi, United Arab Emirates}
\email{sharwan.tiwari@tii.ae}

\thanks{
The first and second authors were supported by the MUR PRIN 2022SC project, Grant No. 2022RFAZCJ.
The first author was also co-funded by the University of Bari through the ``Fondo acquisto e manutenzione
attrezzature per la ricerca'', Grant No. DR 3191.
}
\subjclass[2020]{Primary 94B35; Secondary 94A60, 11T71, 13P10.}

\keywords{Syndrome decoding; Information set decoding; Gr\"obner bases; 
Hamming varieties; Polynomial systems over finite fields.}

\begin{abstract}
We investigate an algebraic approach to the Syndrome Decoding Problem, based on a reformulation
of the Hamming weight constraint and its integration with the Information Set Decoding paradigm.
We begin with a systematic analysis of the Hamming variety, deriving its defining equations
in terms of elementary symmetric functions. Since these equations may have high degree,
we exploit convolution identities for elementary symmetric functions, together with factorizations
based on Lucas' identity, to derive an equivalent formulation with auxiliary variables and
equations of bounded degree.

Building on this modeling, we generalize the ISD paradigm through an ISD-like decoding strategy,
implemented by the \GBDecode\ algorithm, in which only a subset of an information set
is fixed. This approach reduces the size of the combinatorial search space at the cost
of solving the associated multivariate nonlinear systems. To handle this algebraic component,
we employ the \MultiSolve\ algorithm, which replaces a single \Gr\ basis computation with
a collection of computations on simpler systems, obtained by exhaustively assigning
a varying number of indeterminates over the finite field. This provides a tunable balance
between combinatorial search and algebraic solving.

We evaluate the resulting approach experimentally on instances of the Syndrome Decoding Problem
for random binary linear codes, using parameters corresponding to the NIST Security Category 1
parameter set of the Classic McEliece cryptosystem.
The experiments assess the feasibility of this combinatorial-algebraic approach and provide
insights into the practical behavior of \Gr\ basis techniques within an ISD-like decoding framework.
\end{abstract}

\maketitle


\section{Introduction}

Random-looking linear codes are widely regarded as difficult to decode. This
observation, already formalized by Berlekamp, McEliece, and van Tilborg
through the NP-completeness of the general decoding problem \cite{BMEVT},
has provided the theoretical foundation for code-based cryptography for almost
fifty years. The basic idea, introduced by McEliece in 1978 \cite{ME} and later
reformulated by Niederreiter in terms of parity-check matrices \cite{Ni},
consists of hiding the structure of a code with an efficient decoding algorithm
so that it appears to be a random linear code. In practice, this is achieved
by multiplying the generator or parity-check matrix by invertible and permutation
matrices, obtaining an equivalent representation of the code that hides
its original structure.

From the perspective of an adversary, the resulting instance is therefore
expected to resemble a generic decoding problem, for which no efficient algorithms
are currently known, even in the presence of a quantum computer. This aspect
has become especially important since Shor's discovery of polynomial-time quantum
algorithms for integer factorization and discrete logarithms \cite{Shor}, which showed
that the security assumptions of RSA and elliptic-curve cryptography can be broken
by quantum computers. As a consequence, interest in post-quantum cryptography has grown
considerably, and code-based cryptographic schemes have emerged as some of the most
promising candidates for long-term deployment. This is reflected in the NIST
post-quantum standardization process, where code-based cryptography has played
a prominent role, with schemes such as Classic McEliece, HQC, and BIKE being
among the most extensively studied candidates \cite{PQC,ClassME,HQC,BIKE}.

The security of code-based cryptographic constructions relies on the
hardness of the Syndrome Decoding Problem (SDP), and in particular of its
exact-weight variant, the Exact Syndrome Decoding Problem (ESDP). Given a
random parity-check matrix $H$ and a syndrome $s$, the goal is to recover an
error vector $e$ of prescribed Hamming weight satisfying $H e = s$.
Understanding exactly how hard this problem is, and for which parameters,
is therefore not just a theoretical question: it directly determines the
security level of cryptosystems used in practice.

For decades, Information Set Decoding (ISD) has been the dominant
approach to attacking the SDP in the generic setting. The original algorithm,
proposed by Prange \cite{Pra}, introduced the fundamental idea underlying all
subsequent developments. One selects a set of free variables, referred to as an
information set, for the linear system corresponding to the matrix equation $H e = s$,
and assumes that the solution is zero on these variables.

Equivalently, one extends the linear system with additional linear equations
imposing these zero constraints, thereby obtaining a linear system with a unique
solution that, under the assumption, has the prescribed weight. The Syndrome Decoding
Problem can therefore be solved using linear algebra alone. Since Prange's assumption
is satisfied only with a certain probability, ISD repeatedly changes the information set
until a suitable one is found. This success probability consequently determines
the combinatorial complexity of the ISD attack.

A major further development came with Stern, who introduced a meet-in-the-middle
technique that substantially reduced the exponential complexity of Prange's ISD
algorithm \cite{Stern}. Dumer subsequently refined this approach using generalized
birthday techniques \cite{Dumer}. The representation techniques introduced by May, Meurer,
and Thomae \cite{MMT}, and later refined by Becker, Joux, May, and Meurer \cite{BJMM},
led to further improvements in asymptotic complexity. More recent advances have incorporated
nearest-neighbor search and sieving techniques \cite{MayOzerov}. Despite these successive
refinements, however, the underlying algorithmic paradigm has remained essentially unchanged:
ISD is still, at its core, an optimized combinatorial search over information sets.
No sub-exponential algorithm is known for the cryptographic parameter regime, and this
remarkable persistence is one of the main reasons why decoding problems remain attractive
as hardness assumptions. At the same time, the relatively stable complexity of ISD
also suggests exploring alternative algorithmic paradigms that may provide a complementary
perspective on the computational structure of decoding.

One such paradigm is provided by algebraic methods. In \cite{MPS} and \cite{CCMMP},
the Hamming weight constraint is shown to admit an algebraic modeling, in which
the coordinates of an arbitrary vector are regarded as variables and constrained by
a system of multivariate nonlinear equations. Thus, the (Exact) Syndrome Decoding Problem
can be formulated as a problem of solving a multivariate polynomial system, in which
the linear equations corresponding to $H e = s$ are supplemented with nonlinear equations
defining the Hamming variety, namely, the set of vectors having a prescribed Hamming weight.
Note that \cite{MPS} actually establishes an equivalence between the Syndrome Decoding Problem
and the Multivariate Quadratic (MQ) problem, while \cite{Zajac} investigates the reduction
of the MQ problem to SDP through MRHS representations.

A first contribution of the present work is a systematic analysis of the Hamming variety,
leading to the derivation of its defining equations in terms of elementary symmetric functions
and to a reduction in their number. This analysis is presented in Section 2. Since these equations
may have high degree, Sections 3--5 develop a degree-reduction strategy based on auxiliary variables,
convolution identities for elementary symmetric functions, and factorizations derived from
Lucas' identity, ultimately yielding an equivalent formulation with equations of bounded degree.

We further extend the ISD approach to this algebraic formulation of the Syndrome
Decoding Problem. In Sections 6 and 7, we investigate the use of subsets of information
sets, at the cost of solving multivariate nonlinear systems, for instance by means
of \Gr\ basis techniques. This approach, implemented in the \GBDecode\ algorithm in Section 7,
reduces the combinatorial complexity of the ISD search, while introducing the additional
cost of solving the resulting nonlinear systems.

To explore this trade-off, we employ the \MultiSolve\ algorithm \cite{LSPTV,LST}, described
in Section 8, which replaces the computation of a single \Gr\ basis by a collection of \Gr\ basis
computations for simpler systems, obtained by exhaustively assigning a varying number
of indeterminates over the finite field. The practical behavior of this strategy within
an ISD-like attack is investigated experimentally in Section 9 on Syndrome Decoding instances
for random binary linear codes, using the NIST Security Category 1 setting of the Classic McEliece
cryptosystem, with $n = 3488, k = 2720, t = 64$.


\section{The Hamming ideals}

Let $\FF_2 = \GF(2)$ denote the binary field and let $n > 0$ be an integer.
Consider the vector space $V = \FF_2^n$. The {\em Hamming weight} of
a vector $v = (v_1,\ldots,v_n)\in V$ is defined as
\[
\wt(v) = \# \{i \mid v_i = 1\}\in\{0,1,\ldots,n\}
\]
Let $t = \wt(v)$ and set $l = \lfloor \log_2(n) \rfloor$. Consider
the binary expansion
\[
t = \sum_{0\leq k\leq l} t_k 2^k\ (t_k\in\FF_2)
\]
A first goal is to express the binary digits $t_k$ as Boolean functions of the coordinates
$v_1,\ldots,v_n$.

Let $\FF$ be any field and let $\FF[x_1,\ldots,x_n]$ denote the algebra
of the multivariate polynomials with coefficients in $\FF$.
For each integer $d\geq 0$, the {\em elementary symmetric function
of degree $d$ in the variables $x_1,\ldots,x_n$} is by definition the polynomial
\[
e_d(x_1,\ldots,x_n) =
\sum_{1 \le i_1 < \cdots < i_d \le n} x_{i_1}\cdots x_{i_d}\in \FF[x_1,\ldots,x_n]
\]
with the convention that $e_0 = 1$ and $e_d = 0$ for $d < 0$ or $d > n$. Throughout the paper,
ESF stands for elementary symmetric function. When $\FF = \mathbb{F}_2$, these are
referred to as Boolean elementary symmetric functions.

\begin{lemma}
\label{lemma1}
Let $v = (v_1,\ldots,v_n)\in V$ and put $t = \wt(v)$. Then
\[
e_d(v_1,\ldots,v_n)\equiv \binom{t}{d}\ \mod\ 2
\]
\end{lemma}

\begin{proof}
By definition, we have
\[
e_d(v_1,\ldots,v_n) = \sum_{1\le i_1<\cdots<i_d\le n} v_{i_1}\cdots v_{i_d}\in\FF_2
\]
Since $v_i\in\FF_2$, a monomial $v_{i_1}\cdots v_{i_d}$ equals $1$ if and only if
all the chosen indices correspond to nonzero coordinates of $v$. Equivalently,
\[
\{i_1,\ldots,i_d\}\subset \{i \mid v_i = 1\} = S
\]
If $t = \wt(v) = \# S$, the number of such $d$-tuples is exactly $\binom{t}{d}$.
\end{proof}

Lucas' theorem (1878) is a well-known result describing the reduction of binomial
coefficients modulo a prime $p$ in terms of their base-$p$ expansions. In the case $p = 2$,
it yields the following.

\begin{theorem}[Lucas' Theorem modulo 2]
Let $t = \sum_k t_k 2^k$ and $d = \sum_k d_k 2^k$ be the binary expansions of two integers
$t,d\geq 0$. Then
\[
\binom{t}{d} \equiv \prod_k \binom{t_k}{d_k}\ \mod\ 2
\]
\end{theorem}

As an immediate application of Lucas' theorem, we obtain the following result.

\begin{theorem}
\label{main1}
Let $v = (v_1,\dots,v_n)\in V$ and let
\[
t = \sum_{0\leq k\leq l}t_k 2^k
\]
be the binary expansion of the Hamming weight $t = \wt(v)$, where $l = \lfloor \log_2(n) \rfloor$
and $t_k\in\FF_2$. Then, for each $k$, one has
\[
t_k = e_{2^k}(v)
\]
where $e_{2^k}$ denotes the Boolean ESF of degree $2^k$.
\end{theorem}

\begin{proof}
By Lucas' Theorem modulo 2, applied for $d = 2^k$ ($0\leq k\leq l$), one has
\[
\binom{t}{2^k}\equiv t_k\ \mod\ 2
\]
Since $e_{2^k}(v) = \binom{t}{2^k}\ \mod\ 2$, the claim follows.
\end{proof}

Based on the previous results, we introduce the following notions.

\begin{definition}
For any integer $0\leq t\leq n$, the {\em Hamming variety of weight $t$}
is defined as
\[
H_t = \{ v \in V \mid \wt(v) = t \}
\]
Let $R = \FF_2[x_1,\ldots,x_n]$ and let $\cF = \langle x_i^2 - x_i \mid 1\leq i\leq n \rangle$
be the {\em field equations ideal} of $R$. Write the binary expansion of $t$ as
\[
t = \sum_{0\leq k\leq l} t_k 2^k\ (l = \lfloor \log_2(n) \rfloor)
\]
Then, the {\em Hamming ideal of weight $t$} is defined as the ideal $I_t + \cF \subset R$
where
\[
I_t = \langle e_{2^k} - t_k \mid 0\leq k\leq l \rangle
\]
\end{definition}

\medskip
Let $J\subset R$ be an ideal and let $\bar{\FF}_2$ denote the algebraic
closure of the field $\FF_2$. If $\bar{V} = \bar{\FF}_2^n$, one defines
\[
V(J) = \{ v\in \bar{V}\mid f(v) = 0, \forall f\in J\},\
V_{\FF_2}(J) = V(J)\cap V.
\]
By a well-known consequence of the Nullstellensatz over finite fields (see,
for instance, \cite{Gh}), one has
\[
V_{\FF_2}(J) = V(J + \cF)
\]
with $J + \cF$ a radical ideal.

Combining these results with Theorem \ref{main1} shows that the Hamming variety
and the Hamming ideal are related by
\[
V_{\FF_2}(I_t) = V(I_t + \cF) = H_t
\]

An interesting question is whether the number of generators of the ideal $I_t + \cF$ can
be reduced. Recall that $V = \FF_2^n$ denotes the $n$-dimensional vector space over the binary
field.

\begin{theorem}
Let $0\leq t\leq n$ and write $t = \sum_{0\leq k\leq l} t_k 2^k$ $(l = \lfloor \log_2(n) \rfloor)$.
For an integer $0\leq L\leq l$, the following properties are equivalent:
\begin{itemize}
\item[$(1)$] $\forall v \in V$, if $\wt(v) = \sum_{0\leq k\leq l} t'_k 2^k$ and
$t'_0 = t_0,\ldots,t'_L = t_L$ then $\wt(v) = t$;
\item[$(2)$] $2^{L+1} > \max(t, n-t)$.
\end{itemize}
In particular, the minimal integer $L$ for which these equivalent properties hold is
\[
L = \lceil \log_2( \max(t,n-t) + 1 ) \rceil - 1.
\]
When $n = 2^m$, this minimal $L$ is equal to $m$ for $t \in \{0,n\}$, and to $m-1$ otherwise.
\end{theorem}

\begin{proof}
Since the Hamming weight $\wt:V \to \{0,1,\dots,n\}$ is a surjective map,
the condition $(1)$ is equivalent to the following arithmetical statement
\begin{itemize}
\vspace{2pt}
\item[$(1')$] if $0\leq t'\leq n$ and $t' \equiv t\ \mod\ 2^{L+1}$ then $t' = t$.
\vspace{2pt}
\end{itemize}
Assume by contradiction that condition $(2)$ does not hold, that is
\[
2^{L+1} \leq \max(t,n-t)
\]
Hence, either $t + 2^{L+1} \leq n$ or $t - 2^{L+1} \geq 0$.
In the first case, let $t' = t + 2^{L+1}$; in the second case, let
$t' = t - 2^{L+1}$. In both cases, we have $0\leq t'\leq n$, $t'\neq t$,
and
\[
t' \equiv t\ \mod\ 2^{L+1}
\]
This contradicts condition $(1')$, which states that $t$ is the unique
integer in $\{0,1,\ldots,n\}$ modulo $2^{L+1}$. Therefore, condition $(1')$ necessarily
implies condition $(2)$.

In a similar way, assume now that condition $(2)$ holds, that is
\[
2^{L+1} > \max(t,n-t)
\]
Then $t - 2^{L+1} < 0$ and $t + 2^{L+1} > n$, so there exists no integer
$0\leq t'\leq n$, $t'\neq t$, such that $t' \equiv t\ \mod\ 2^{L+1}$.
Thus $(1')$ holds, and consequently condition $(1)$ is satisfied.

Finally, condition $(2)$ is equivalent to
\[
L+1 > \log_2( \max(t,n-t) )
\]
and the minimal integer $L$ for which this inequality holds is
\[
L = \lceil \log_2(\max(t,n-t)+1) \rceil - 1
\]
The final claim for the case $n = 2^m$ follows immediately.
\end{proof}

Using the above result, the generators of the Hamming ideal corresponding
to the elementary symmetric polynomials of degree $2^k$ with $k > L$ can be
omitted modulo $\cF$. Hence, the Hamming ideal $I_t + \cF$ can equivalently
be defined by
\[
I_t = \langle e_{2^k} - t_k \mid 0\leq k\leq L \rangle
\]
where $L = \lceil \log_2( \max(t,n-t) + 1 ) \rceil - 1$ and $t = \sum_{0\leq k\leq L} t_k 2^k$
is the binary representation of the weight $0\leq t\leq n$.

\section{Convolution Hamming ideals}

As $n$ increases, the explicit computation of the polynomial $e_{2^k}$ rapidly becomes
infeasible, due to the presence of $\binom{n}{2^k}$ monomials.
Therefore, we look for an alternative characterization of the Hamming ideal $I_t + \cF$
as the elimination ideal of a larger lifted ideal. In this context, recursive definitions
of the elementary symmetric functions provide a natural approach. We start with combinatorial
identities that hold over an arbitrary field $\FF$.

To simplify the notation, we introduce the following convention. For any integers $d\geq 0$
and $1\leq a\leq b\leq n$, we set
\[
e_d^{(a,b)} = e_d(x_a,\ldots,x_b) \in \FF[x_1,\ldots,x_n]
\]
where $e_d(x_a,\ldots,x_b)$ denotes the elementary symmetric polynomial of degree $d$
in the variables $x_a,\ldots,x_b$.

We have the following well-known combinatorial identity (see, for instance, \cite{Ma})

\begin{proposition}[Convolution identity]
Let $n\geq 2$, $0\leq d\leq n$ and $1\leq m < n$ be integers. It holds
\[
e_d^{(1,n)} = \sum_{0\leq k\leq d} e_k^{(1,m)} e_{d-k}^{(m+1,n)}
\]
where we adopt the convention that $e_0^{(a,b)} = 1$ and $e_k^{(a,b)} = 0$ for $k > b-a+1$.
Equivalently, the sum can be restricted to
\[
\max(0,d-(n-m)) \leq k \leq \min(d,m)
\]
\end{proposition}

As a special case, taking $m = n - 1$ yields the classical Stifel identity.

\begin{proposition}[Stifel identity]
Let $0\leq d\leq n$. Then
\[
e_d^{(1,n)} = e_d^{(1,n-1)} + x_n e_{d-1}^{(1,n-1)}
\]
\end{proposition}

In contrast to Stifel’s choice $m = n - 1$, one may take $m = \lfloor (1+n)/2 \rfloor =
\lceil n/2 \rceil$, in which case $n - m = \lfloor n/2 \rfloor$, yielding a balanced
partition of the variables.

We consider the construction of the elementary symmetric functions appearing in the definition
of the Hamming ideal $I_t + \cF$ ($0\leq t\leq n$), namely $e_{2^k}^{(1,n)}$ for all $0\leq k\leq L$,
where
\[
L = \lceil \log_2( \max(t,n-t)+1 ) \rceil - 1
\]
We seek to bound the number of ESFs produced by the recursive convolution scheme
obtained by splitting each set of variables into two balanced blocks. The starting point
is the following lemma.

\begin{lemma}
\label{lemma2}
Let $n\geq 2$, and let $T_n$ be the set of intervals generated by the recursive binary
decomposition of $\{1,\ldots,n\}$, where each interval $(a,b) = \{a,\ldots,b\}$ $(a < b)$
is split into
\[
(a,m)\cup (m+1,b)
\]
with $m = \lfloor \frac{a+b}{2} \rfloor$. Let $0\leq L\leq \lfloor \log_2(n) \rfloor$
be an integer, and define
\[
N(n,L) = \sum_{(a,b)\in T_n} \min(|(a,b)|,2^L)
\]
where $|(a,b)| = b-a+1$ denotes the length of the interval $(a,b)$. Then
\[
N(n,L) = \cO(n L)
\]
\end{lemma}

\begin{proof}
Let $D$ be the depth of the decomposition binary tree $T_n$, so that
\[
D = \lceil \log_2(n) \rceil
\]
For each $0\leq k\leq D$, let $Z_k$ denote the set of intervals at level $k$ of the decomposition.
We observe two structural properties. First, note that $|Z_k|\leq 2^k$ for all $k$. Second, for each
fixed $k$, the intervals in $Z_k$ are pairwise disjoint and contained in $\{1,\dots,n\}$, hence
\[
\sum_{I\in Z_k} |I| \leq n
\]
We write
\[
N(n,L) = \sum_{0\leq k\leq D}\sum_{I\in Z_k} \min(|I|,2^L)
\]
and split the sum into two ranges.

Case 1: $0\leq k\leq D-L$. Using $\min(|I|,2^L)\leq 2^L$ and $|Z_k|\leq 2^k$, we obtain
\[
\begin{array}{l}
\sum_{0\leq k\leq D-L} \sum_{I\in Z_k} \min(|I|,2^L) \leq
(\sum_{0\leq k\leq D-L} 2^k)\, 2^L \\ [6pt]
\qquad = (2^{D-L+1}-1)\, 2^L = 2^{D+1} - 2^L < 2^{D+1}
\end{array}
\]
Since $D = \lceil \log_2(n) \rceil$, we have $n\leq 2^D\leq 2 n$, and consequently
$2^{D+1} \leq 4 n$. Therefore, it holds
\[
\sum_{0\leq k\leq D-L} \sum_{I\in Z_k} \min(|I|,2^L)\leq 4 n
\]

Case 2: $D-L+1\leq k\leq D$. Since $\min(|I|,2^L)\leq |I|$, we have
\[
\sum_{I\in Z_k} \min(|I|,2^L) \leq \sum_{I\in Z_k}|I| \leq n
\]
There are at most $L$ such levels, so
\[
\sum_{D-L+1\leq k\leq D} \sum_{I\in Z_k} \min(|I|,2^L) \leq L n
\]
Combining the two estimates gives
\[
N(n,L)\leq 4 n + L n = n(L + 4)
\]
which proves the claim.
\end{proof}

An immediate consequence of the above result is the following proposition.

\begin{proposition}
\label{main2}
Let $n\geq 2$ and $0\leq L\leq \lfloor \log_2(n) \rfloor$ be integers. Consider the recursive
construction of the elementary symmetric functions $e_{2^k}^{(1,n)}$ for $0\leq k \leq L$,
obtained by repeated application of the convolution identity with balanced splits. Then,
the total number of ESF instances generated by this recursion is
\[
N(n,L) = \sum_{(a,b)\in T_n} \min(|(a,b)|,2^L)
\]
where the sum ranges over all intervals $(a,b)$ in the recursive binary decomposition tree $T_n$.
Moreover, it holds
\[
N(n,L) = \cO(n L)
\]
\end{proposition}

\begin{proof}
For each interval $(a,b)\in T_n$, the convolution formula for $e_{2^k}^{(a,b)}$
involves intermediate degrees $0\leq d\leq 2^k$. Since $2^k\leq 2^L$ for every
$0\leq k\leq L$, the recursive computation at node $(a,b)$ requires all ESFs
\[
e_d^{(a,b)}\ (1\leq d\leq \min(|(a,b)|,2^L))
\]
Thus, excluding the trivial constant term $e_0^{(a,b)} = 1$, the number of ESF
instances generated at node $(a,b)$ is $\min(|(a,b)|,2^L)$.
Summing over all nodes gives
\[
N(n,L) = \sum_{(a,b)\in T_n}\min(|(a,b)|,2^L)
\]
The bound $N(n,L) = \cO(nL)$ follows directly from the previous lemma.
\end{proof}

We now translate the recursive construction from Proposition \ref{main2} into
an algebraic setting.

\begin{theorem}
\label{main3}
Let $0\leq t\leq n$ be an integer and let $t = \sum_{0\leq k\leq L} t_k 2^k$
be its binary expansion where $L = \lceil \log_2( \max(t,n-t) + 1 ) \rceil - 1$.
Consider the two sets of variables
\[
\begin{array}{c@{\ }c@{\ }l}
X & = & \{x_1,\ldots,x_n\}, \\
Y & = & \{ y_d^{(a,b)}\, |\, 1\leq d\leq \min(|(a,b)|,2^L), (a,b)\in T_n \}
\end{array}
\]
Let $R = \FF_2[X]$ and $R' = \FF_2[X \cup Y]$ be the corresponding polynomial algebras,
and let $\cF$ and $\cF'$ denote their respective field equations ideals.
Define the following sets of polynomials
\[
\begin{array}{c@{\ }c@{\ }l}
A & = & \{ y_1^{(a,a)} - x_a \mid 1\leq a\leq n \}, \\
B & = & \{ y_d^{(a,b)} - \sum_{0\leq k\leq d} y_k^{(a,m)} y_{d-k}^{(m+1,b)} \mid
        1\leq d\leq \min(b-a+1,2^L), \\
  &   & (a,b)\in T_n, a < b, m = \lfloor \frac{a+b}{2} \rfloor \}, \\
C & = & \{ y_{2^k}^{(1,n)} - t_k \mid 0\leq k \leq L \}
\end{array}
\]
For polynomials in $B$, we adopt the convention that $y_0^{(a,b)} = 1$ and $y_d^{(a,b)} = 0$
for $d > b-a+1$. Moreover, we put
\[
N = N(n,L) = \sum_{(a,b)\in T_n} \min(b-a+1,2^L)
\]
By considering the ideal $J_t = \langle A \cup B \cup C \rangle \subset R'$, we obtain
\[
(J_t + \cF')\cap R = I_t + \cF,\, \varphi( V_{\FF_2}(J_t) ) = V_{\FF_2}(I_t) = H_t
\]
where $\varphi:\FF_2^{n + N} \to \FF_2^n$ denotes the canonical projection onto
the $X$-coordinates.
\end{theorem}

\begin{proof}
By construction, the polynomials in $B$ encode the balanced convolution identities,
so recursively
\[
y_d^{(a,b)} = e_d^{(a,b)}
\]
for all $1\le d\le \min(b-a+1,2^L)$ and every $(a,b)\in T_n$. The polynomials in $C$
impose
\[
e_{2^k}^{(1,n)} = t_k\ (0\le k\le L)
\]
according to the binary expansion of $t$. Hence $I_t + \cF\subseteq (J_t + \cF')\cap R$.

Conversely, each auxiliary variable $y_d^{(a,b)}$ is uniquely determined by the
relations in $A\cup B$, recursively along the binary decomposition tree $T_n$.
More precisely, the variables attached to leaves are fixed by $A$, while those attached to
internal intervals are determined recursively from previously defined variables.
Thus, adjoining the variables in $Y$ introduces no new algebraic relations
among the variables in $X$, and therefore
\[
(J_t + \cF')\cap R = I_t + \cF.
\]
\end{proof}

We call the ideal $J_t + \cF'$ defined in Theorem \ref{main3} the {\em Convolution Hamming ideal
of weight $t$}. Since the elementary symmetric functions $e_{2^k}^{(1,n)}$ are the only
ones appearing in the Hamming ideal $I_t + \cF$, one can use Lucas' Theorem to factor
any ESF in terms of ESFs of degree a power of 2, modulo the field-equation ideal $\cF$.

\section{Factorized Convolution Hamming ideals}

An immediate consequence of Lucas' Theorem modulo 2 is the following factorization identity
among Boolean elementary symmetric functions.

\begin{theorem}[Factorization identity]
Let $0\leq d\leq n$ and let $d = \sum_k d_k 2^k$ be its binary expansion. Then
\[
e_d(x_1,\ldots,x_n)\equiv \prod_{k\mid d_k = 1} e_{2^k}(x_1,\ldots,x_n)\ \mod\ \cF
\]
\end{theorem}

\begin{proof}
Let $v = (v_1,\ldots,v_n)\in V$ and put $t = \wt(v)$. By Lemma \ref{lemma1}, we have
\[
e_d(v)\equiv \binom{t}{d}\ \mod\ 2
\]
Moreover, by Lucas' theorem
\[
\binom{t}{d} \equiv \prod_k \binom{t_k}{d_k}\ \mod\ 2
\]
where $t = \sum_k t_k 2^k$. We conclude that
\[
e_d(v) = 1 \iff t_k\geq d_k,\ \forall\ k \iff d_k = 1\ \text{implies that}\ t_k = 1,\ \forall\ k
\]
By using again Lemma \ref{lemma1} and Lucas' Theorem, we obtain
\[
e_{2^k}(v) \equiv \binom{t}{2^k}\ \mod\ 2\equiv t_k\ \mod\ 2
\]
It follows that
\[
e_{2^k}(v) = 1 \iff t_k = 1
\]
Let us consider the product
\[
\prod_{k\mid d_k = 1} e_{2^k}(v)
\]
This product is equal to one if and only if $e_{2^k}(v) = 1$ for all $k$ such that $d_k = 1$,
or equivalently $t_k = 1$ for all such $k$. We conclude that
\[
e_d(v) = 1 \iff \prod_{k\mid d_k = 1} e_{2^k}(v) = 1
\]
In other words, we have that $e_d(v) = \prod_{k\mid d_k = 1} e_{2^k}(v)$, for all $v\in V$ and
therefore
\[
e_d(x_1,\ldots,x_n)\equiv \prod_{k\mid d_k = 1} e_{2^k}(x_1,\ldots,x_n)\ \mod\ \cF
\]
\end{proof}

Applying the factorization identity to the convolution identity yields the following result.

\begin{theorem}[Factorized Convolution identity]
Let $n\geq 2$ and $1\leq m < n$ be integers. Let $0\leq d\leq n$. The following identity holds
\[
e_d^{(1,n)} \equiv
\sum_{0\leq j\leq d}
( \prod_{k\mid j_k = 1} e_{2^k}^{(1,m)} )
( \prod_{k\mid (d - j)_k = 1} e_{2^k}^{(m+1,n)} )
\ \mod\ \cF
\]
where $j = \sum_k j_k 2^k$ and $d - j = \sum_k (d - j)_k 2^k$ are the binary expansions
of $j$ and $d-j$, respectively.
\end{theorem}

We now estimate the number of ESFs arising in the factorized convolution scheme obtained by
recursively partitioning the variables into two balanced blocks. The following combinatorial
lemma provides the key bound.

\begin{lemma}
\label{lemma3}
Let $n\geq 2$, and let $T_n$ be the set of intervals in the recursive binary decomposition
tree of $\{1,\ldots,n\}$ defined in Lemma \ref{lemma2}.
Let $0\leq L\leq \lfloor \log_2(n) \rfloor$. For each interval $I\in T_n$, denote by $|I|$
its length and define
\[
\rho(I,L) = \# \{k\geq 0\mid 2^k\leq \min(|I|,2^L)\} =
\min(\lfloor \log_2(|I|) \rfloor, L) + 1
\]
Set $M(n,L) = \sum_{I\in T_n} \rho(I,L)$. Then
\[
M(n,L) = \cO(n)
\]
uniformly in $L$, that is, there exists an absolute constant $C > 0$ such that
\[
M(n,L)\leq C n
\]
for all $n\geq 2$ and all $0\leq L\leq \lfloor \log_2(n) \rfloor$.
\end{lemma}

\begin{proof}
Recall that the set $T_n$ is constructed recursively by splitting each interval $I = (a,b)$ ($a < b$)
into
\[
(a,m) \cup (m+1,b)
\]
where $m = \lfloor \frac{a+b}{2} \rfloor$. Hence every internal node has exactly two children,
and the depth of the tree is $D = \lceil \log_2(n) \rceil$. For every interval $I\in T_n$, we have
\[
\rho(I,L) = \min(\lfloor \log_2(|I|)\rfloor, L) + 1\leq \log_2(|I|) + 1
\]
Therefore, it holds
\[
M(n,L)\leq \sum_{I\in T_n} (\log_2(|I|) + 1)
\]
For each depth $0\leq d\leq D$, let $Z_d$ denote the set of intervals at depth $d$.
Since every node has at most two children, we have
\[
|Z_d|\leq 2^d
\]
Moreover, by construction of the balanced decomposition, every interval $I\in Z_d$ satisfies
\[
|I|\leq \left\lceil \frac{n}{2^d} \right\rceil \leq \frac{n}{2^d} + 1
\]
Hence
\[
\log_2(|I|) \leq \log_2(\frac{n}{2^d} + 1)
\]
Since $D = \lceil \log_2(n) \rceil$, we have $n\leq 2^D$ and therefore
$\frac{n}{2^d} + 1\leq 2^{D-d} + 1$.

Because $d\leq D$, we have $2^{D-d}\geq 1$, hence $2^{D-d} + 1\leq 2^{D-d+1}$. We conclude
\[
\log_2(|I|)\leq D-d+1
\]
Summing over all intervals at depth $d$, we obtain
\[
\sum_{I\in Z_d} (\log_2(|I|) + 1)\leq 2^d(D-d+2)
\]
Summing over all depths,
\[
M(n,L)\leq \sum_{0\leq d\leq D} 2^d(D-d+2)
\]
Factoring out $2^D$, we have
\[
M(n,L)\leq 2^D\cdot \sum_{0\leq d\leq D} \frac{D-d+2}{2^{D-d}}
\]
Setting $r = D - d$, we obtain
\[
M(n,L)\leq 2^D\cdot \sum_{0\leq r\leq D} \frac{r+2}{2^r}
\]
Since
\[
\sum_{r=0}^{\infty}\frac{r+2}{2^r} =
\sum_{r=0}^{\infty}\frac{r}{2^r} + 2\sum_{r=0}^{\infty}\frac1{2^r} = 2 + 4 = 6
\]
we conclude that
\[
M(n,L)\leq 6\cdot 2^D
\]
Finally, since $2^D\leq 2 n$, we obtain
\[
M(n,L)\leq 12 n
\]
\end{proof}

An immediate consequence of Lemma \ref{lemma3} and the factorized convolution identity is
the following proposition.

\begin{proposition}
\label{main4}
Let $n\geq 2$ and $0\leq L\leq \lfloor \log_2(n) \rfloor$ be integers. Consider the recursive
construction of the elementary symmetric functions $e_{2^k}^{(1,n)}$ for $0\leq k \leq L$
obtained by repeated application of the factorized convolution identity along the recursive
binary decomposition tree $T_n$. Let $M(n,L)$ denote the total number of ESF instances
generated by this construction. Then
\[
M(n,L) = \sum_{I\in T_n} \rho(I,L)
\]
where $\rho(I,L) = \# \{k\geq 0 \mid 2^k \leq \min(|I|,2^L)\} =
\min(\lfloor \log_2(|I|) \rfloor, L) + 1$. By Lemma \ref{lemma3}, it follows that
\[
M(n,L) = \cO(n)
\]
uniformly in $L$.
\end{proposition}

\begin{proof}
For each interval $I\in T_n$, the factorized convolution formula for $e_{2^k}^I$
decomposes the computation into contributions from the two balanced subintervals of $I$.
Unlike the standard convolution, the factorized version ensures that, at each recursive step,
only contributions compatible with powers of two degrees are propagated through the tree.

Since $0\leq k\leq L$, the recursion at node $I$ generates ESF instances corresponding to all
levels $k$ such that the degree $2^k$ is admissible with respect to both $|I|$ and the
truncation degree $2^L$. The number of such levels is therefore
\[
\rho(I,L) = \# \{k\geq 0 \mid 2^k \leq \min(|I|,2^L)\}
\]
Summing over all intervals $I\in T_n$, we obtain
\[
M(n,L) = \sum_{I\in T_n} \rho(I,L)
\]
The bound $M(n,L) = \cO(n)$ follows directly from Lemma \ref{lemma3}.
\end{proof}

Using the complexity compression provided by Lemma \ref{lemma3} applied to the recursive
construction of Proposition \ref{main4}, we obtain a compact version of a Convolution
Hamming ideal. The proof follows the same arguments as in Theorem \ref{main3}.

\begin{theorem}
\label{main5}
Let $0 \leq t \leq n$ be an integer and let $t = \sum_{0\leq k\leq L} t_k 2^k$
be its binary expansion, where $L = \lceil\log_2( \max(t,n-t) + 1)\rceil - 1$.
Let $T_n$ be the recursive binary decomposition tree of $\{1,\ldots,n\}$.
For each interval $I = (a,b)\in T_n$, we define
\[
\rho(I,L) = \min(\lfloor \log_2 |I| \rfloor, L) + 1
\]
Consider the sets of variables
\[
\begin{array}{c@{\ }c@{\ }l}
X & = & \{x_1,\ldots,x_n\}, \\
Y & = & \{y_{2^k}^{(a,b)} \mid (a,b)\in T_n,\ 0\leq k < \rho((a,b),L)\}
\end{array}
\]
Let $R = \FF_2[X]$ and $R' = \FF_2[X \cup Y]$, and define the following
sets of polynomials
\[
\begin{array}{c@{\ }c@{\ }l}
A & = & \{ y_1^{(a,a)} - x_a \mid 1 \leq a \leq n \}, \\
B & = & \{ y_{2^k}^{(a,b)} - f_{2^k}^{(a,b)} \mid
(a,b)\in T_n,\ a < b,\ 0 \leq k < \rho((a,b),L) \}, \\
C & = & \{ y_{2^k}^{(1,n)} - t_k \mid 0 \leq k \leq L \}
\end{array}
\]
where
\[
f_{2^k}^{(a,b)} =
\sum_{0\leq j\leq 2^k}
( \prod_{h\mid j_h = 1} y_{2^h}^{(a,m)} )
( \prod_{h\mid (2^k - j)_h = 1} y_{2^h}^{(m+1,b)} )
\]
with $m = \lfloor \frac{a+b}{2} \rfloor$ and $j = \sum_h j_h 2^h,\,
2^k - j = \sum_h (2^k - j)_h 2^h$ denoting the binary expansions of the integers
$j,\, 2^k-j$, respectively.
As usual, we adopt the convention that $y_d^{(a,b)} = 0$ whenever $d > b-a+1$.
Finally, put
\[
M = M(n,L) = \sum_{I\in T_n} \rho(I,L)
\]
and define the ideal $J_t = \langle A \cup B \cup C \rangle \subset R'$.
We have
\[
(J_t + \cF')\cap R = I_t + \cF,\,
\varphi( V_{\FF_2}(J_t) ) = V_{\FF_2}(I_t) = H_t
\]
where $\varphi:\FF_2^{n + M} \to \FF_2^n$ denotes the canonical projection onto
the $X$-coordinates.
\end{theorem}

We refer to the ideal $J_t + \cF'$ constructed in Theorem \ref{main5} as the {\em Factorized
Convolution Hamming ideal of weight t}.

\section{Quadratic Factorized Convolution Hamming ideals}

A straightforward consequence of the factorized convolution formula is the following result.

\begin{corollary}[Quadratic Factorized Convolution identities]
Let $n\geq 2$ and $1\leq m < n$. For every subset $S\subset\NN$, we define
polynomials in $R = \FF_2[x_1,\ldots,x_n]$ by
\[
c^{(1,m)}_S = \prod_{k\in S} e_{2^k}^{(1,m)}
\]
with the convention that $c^{(1,m)}_\emptyset = 1$. Equivalently, the polynomials
$c^{(1,m)}_S$ are defined recursively by
\[
c^{(1,m)}_{S\cup\{k\}} = c^{(1,m)}_S e_{2^k}^{(1,m)}\ (k\notin S)
\]
Similarly, we define polynomials $c^{(m+1,n)}_S\in R$. Then,
the elementary symmetric function of degree $d\geq 0$ satisfies
the following quadratic identity
\[
e_{d}^{(1,n)} \equiv \sum_{0\leq j\leq d} c^{(1,m)}_{\{k\mid j_k = 1\}}
c^{(m+1,n)}_{\{k\mid (d-j)_k = 1\}}
\ \mod\ \cF
\]
\end{corollary}

We now study the number of non-constant polynomials $c_S^{(a,b)}$ ($S\neq\emptyset$)
arising in the quadratic factorized convolution scheme obtained by recursively
partitioning the intervals into two balanced blocks. A quantitative bound is provided
by the following result.

\begin{lemma}
\label{lemma4}
Let $n\geq 2$, and consider $T_n$, the set of intervals in the recursive binary
decomposition tree of $\{1,\ldots,n\}$ defined in Lemma \ref{lemma2}.
Let also $0 < L\leq \lfloor \log_2(n) \rfloor$. For each interval $I\in T_n$, define
as in Lemma \ref{lemma3}
\[
\rho(I,L) = \min(\lfloor \log_2(|I|)\rfloor, L) + 1
\]
By defining
\[
C(n,L) = \sum_{I\in T_n} (2^{\rho(I,L)} - 1)
\]
we have that $C(n,L) = \cO( n L )$.
\end{lemma}

\begin{proof}
The set $T_n$ is obtained from a recursive binary decomposition of $\{1,\ldots,n\}$,
where each interval $I = (a,b)$ ($a < b$) is split as $I = (a,m) \cup (m+1,b)$ with
$m = \lfloor \frac{a+b}{2} \rfloor$.  Hence, the resulting tree has depth $D = \lceil \log_2(n) \rceil$,
and at each depth $0\leq d\leq D$ there are at most $2^d$ intervals.

For an interval $I\in T_n$ at depth $d$, as in the Lemma \ref{lemma3}
we obtain
\[
\lfloor \log_2(|I|) \rfloor \leq D - d + 1
\]
Consequently,
\[
\rho(I,L) = \min(\lfloor \log_2(|I|)\rfloor, L) + 1 \leq \min(D-d+1, L) + 1
\]

We aim to estimate
\[
C(n,L) = \sum_{I\in T_n} (2^{\rho(I,L)} - 1)
\leq \sum_{I\in T_n} 2^{\rho(I,L)}
\]
We split the sum by depth $d$ and denote by $Z_d$ the set of intervals at depth $d$.
Then $|Z_d|\leq 2^d$, and we have
\[
\begin{aligned}
C(n,L) &\leq \sum_{0\leq d\leq D} \sum_{I\in Z_d} 2^{\rho(I,L)}
\leq \sum_{0\leq d\leq D} 2^d \cdot 2^{\min(D-d+1, L) + 1} \\
      &\leq \sum_{0\leq d\leq D} 2^d \cdot 2^{\min(D-d+2, L+1)}
\end{aligned}
\]
We distinguish two cases.

Case 1: $d \le D-L$. Then $D - d + 2\geq L + 2$, hence $\min(D - d + 2,  L + 1) = L + 1$, and
\[
\sum_{I\in Z_d} 2^{\rho(I,L)}\leq 2^d \cdot 2^{L+1}
\]
Summing over $0\leq d\leq D-L$, we obtain
\[
\sum_{0\leq d\leq D-L} 2^d 2^{L+1} \leq 2^{L+1}\cdot \sum_{0\leq d\leq D-L} 2^d
\leq 2^{L+1} \cdot 2^{D-L+1} = 2^{D+2} \leq 8n
\]

Case 2: $d > D-L$. Let $r = D - d$. Then $0 \leq r < L$, and
\[
\min(D - d + 2, L + 1) \leq D - d + 2 = r + 2
\]
Thus, since $d + r = D$
\[
\sum_{I\in Z_d} 2^{\rho(I,L)} \leq 2^d \cdot 2^{r+2} = 2^{D+2}
\]
Summing over the $L$ deepest levels,
\[
\sum_{D-L+1\leq d\leq D} 2^{D+2} \leq L \cdot 2^{D+2} \leq 8 n L
\]

Combining both cases, we obtain
\[
C(n,L) \leq 8n + 8nL = 8 n (L + 1)
\]
which proves the claim.
\end{proof}

Combining Lemma \ref{lemma4} and Proposition \ref{main4}, we derive the following result
on the recursive application of quadratic factorized convolution identities.

\begin{proposition}
\label{main6}
Let $n\geq 2$ and $0 < L\leq \lfloor \log_2(n) \rfloor$ be integers. Consider the recursive
construction of the elementary symmetric functions $e_{2^k}^{(1,n)}$ for $0\leq k \leq L$
obtained by repeated application of the quadratic factorized convolution scheme along
the recursive binary decomposition tree $T_n$. Let $M'(n,L)$ denote the total number of
elementary symmetric functions $e_{2^k}^I$ together with non-constant polynomials
$c_S^I$ $(I\in T_n, S\subset \{0,\ldots,\rho(I,L)-1\})$ generated by this construction.
Then
\[
M'(n,L) = \cO(n L)
\]
\end{proposition}

\begin{proof}
The number of elementary symmetric functions $e_{2^k}^I$ generated in the construction
has already been shown to be $\cO(n)$ in Proposition \ref{main4}.

We now consider the auxiliary non-constant polynomials $c_S^I$. In the recursive application
of the quadratic factorized convolution along the binary decomposition tree $T_n$, for each
interval $I\in T_n$, we have defined
\[
\rho(I,L) = \# \{k\geq 0 \mid 2^k \leq \min(|I|,2^L)\} = \min(\lfloor \log_2(|I|)\rfloor, L) + 1
\]

As a consequence, only elementary symmetric functions of degrees $2^k$ with $k < \rho(I,L)$
can appear in the Lucas factorization at the node $I$. Hence, each non-constant
polynomial $c_S^I$ is indexed by a non-empty subset
\[
S\subset \{0,\ldots,\rho(I,L)-1\}
\]

Using the estimate of Lemma \ref{lemma4} for the combinatorial number
\[
C(n,L) = \sum_{I\in T_n} (2^{\rho(I,L)} - 1)
\]
we obtain that the total number of non-constant polynomials $c_S^I$ is in $\cO(nL)$.
Combining both contributions yields $M'(n,L) = \cO(nL)$.
\end{proof}

For each weight $0 \leq t \leq n$, as in Theorem \ref{main3} and Theorem \ref{main5},
we encode the recursion based on the quadratic factorized convolution identities along
the balanced binary tree $T_n$ as a lifted Hamming ideal, which we call the
{\em Quadratic Factorized Convolution Hamming ideal of weight $t$}.

The number of auxiliary variables required to define such an ideal is $\cO(n L)$,
by Proposition \ref{main6}, and all its generators have degree at most 2.
Note that, although the Factorized Convolution Hamming ideal requires fewer
auxiliary variables, its generators have higher degree.

We will refer to the Convolution, Factorized Convolution, and Quadratic Factorized
Convolution Hamming ideals as {\em C-Hamming}, {\em FC-Hamming}, and
{\em QFC-Hamming} ideals, respectively. In our computational experiments,
we will compare the performance of these schemes on random binary linear codes
with cryptographic parameters.

\section{Syndrome Decoding Problem and Information Sets}

Let $0\leq k\leq n$ and consider the binary vector spaces $V = \FF_2^n$ and $W = \FF_2^{n-k}$.
A {\em binary linear code of dimension $k$ and length $n$} is by definition a $k$-dimensional
subspace $\cC\subset V$. Any such code is defined as $\cC = \Ker(L)$, where $L:V\to W$ is a linear
map of maximal rank. We call $L$ a {\em parity-check mapping} of the code $\cC$. Given a vector
$v\in V$, we define the {\em syndrome of $v$} as the image $s = L(v)\in W$.
Consider the polynomial algebra $R = \FF_2[x_1,\ldots,x_n]$, and denote by
$l_1,\ldots,l_{n-k}\in R$ the $n-k$ independent linear forms defining $L$.
Let $s = (s_1,\ldots,s_{n-k})\in W$ be a syndrome. We define the linear ideal
\[
L_s = \langle l_1 - s_1, \ldots, l_{n-k} - s_{n-k} \rangle\subset R 
\]
We call $L_s$ the {\em syndrome ideal} corresponding to $s$. Note that $V_{\FF_2}(L_s) = L^{-1}(s)$.

For any weight $0\leq t\leq n$, let $I_t + \cF$ be the Hamming ideal of weight $t$, so that
$V_{\FF_2}(I_t) = H_t$. The {\em syndrome decoding ideal} of syndrome $s$ and weight $t$
is defined as the sum
\[
J_{s,t} = L_s + I_t
\]
It follows that $H_{s,t} = V_{\FF_2}(J_{s,t}) = L^{-1}(s)\cap H_t$. We call $H_{s,t}$ the
{\em syndrome decoding variety} corresponding to $s$ and $t$.

The Syndrome Decoding Problem, in its exact and bounded-weight variants, is a fundamental computational
problem in coding theory and code-based cryptography. In particular, the security of the McEliece
cryptosystem and several of its variants relies on the presumed hardness of recovering a vector $v\in V$
of weight $\wt(v)\leq t$ from its syndrome $L(v) = s$. In this paper, we focus on the exact variant
of the Syndrome Decoding Problem, that is, $\wt(v) =  t$ since the exact and bounded-weight formulations
are polynomially equivalent.

Considering the system of polynomial equations corresponding to the vector equations $L(v) = s$
and $\wt(v) =  t$, one has that solving the Exact Syndrome Decoding Problem, briefly ESDP,
amounts to computing some element of the variety $H_{s,t}$.

By definition, the {\em minimum distance} of the code $\cC$ is
\[
d = \min\{\wt(u) \mid u \in \cC, u \neq 0\}
\]
It follows immediately from the definition of the minimum distance that two
distinct error vectors of weight at most
\[
t = \lfloor (d-1)/2 \rfloor
\]
cannot have the same syndrome. The integer $t$ is called the {\em error-correcting
capability} of the code $\cC$. In what follows, {\em we assume} that the parameter $t$
defining the syndrome variety $H_{s,t}$ is the error-correcting capability of $\cC$.
Therefore, for every syndrome $s\in W$,
\[
\# H_{s,t}\leq 1
\]

In the McEliece cryptosystem, encryption with respect to the code $\cC$ is performed by
mapping $u\in C \mapsto w = u + v\in V$, where $v\in H_{s,t}$ and $s = L(w)$.
An unauthorized decryption $u = w - v$ thus amounts to solving the ESDP instance
associated with syndrome $s$.

A fundamental concept in the complexity analysis of the Syndrome Decoding Problem
is the notion of Information Set. With respect to a fixed variable ordering, we
apply Gaussian elimination, or equivalently compute a \Gr\ basis, to the code equations
\[
l_1 = 0, \ldots, l_{n-k} = 0
\]
This yields a partition of the variable set $\{x_1,\ldots,x_n\}$ into a set of
$n-k$ pivot variables and a set of $k$ free variables. With respect to the chosen
variable ordering, an {\em Information Set} of the code $\cC$ is by definition the set
of free variables. Different variable orderings generally yield different Information
Sets. Given an Information Set of the code $\cC$, an {\em Evaluation Set} is any subset
of it.

To solve an ESDP instance, namely to compute $H_{s,t} = V_{\FF_2}(J_{s,t})$, one may
compute a \Gr\ basis $G$ of the ideal $J_{s,t} + \cF$
where $\cF = \langle x_1^2 - x_1, \ldots, x_n^2 - x_n \rangle$ is the field equations
ideal of $R$. Indeed, since we assume that $\# H_{s,t}\leq 1$, for any monomial
ordering of $R$ we have
\[
G = 
\left\{
\begin{array}{cl}
\{x_1 - v_1,\ldots,x_n - v_n\} & \mbox{if}\ H_{s,t} = \{(v_1,\ldots,v_n)\}, \\
\{1\} & \mbox{if}\ H_{s,t} = \emptyset
\end{array}
\right.
\]

To compute a \Gr\ basis of the ideal $J_{s,t} + \cF$ is, in general, a computationally demanding
task for parameter sets arising in practical code-based cryptography. Since
\[
J_{s,t} = L_s + I_t
\]
one problem is that the Hamming ideal $I_t + \cF$ is generated by high-degree polynomials, which may
significantly increase the \Gr\ basis solving degree. A possible way to mitigate this issue
is to replace $I_t + \cF$ with a suitable lifted Hamming ideal, such as C-Hamming, FC-Hamming,
and QFC-Hamming ideals, introduced in the previous sections.

A second challenge is the potentially large number of variables involved.
Fix an Evaluation Set and, for ease of notation, assume that it is $\{x_1,\ldots,x_r\}$.
Let $u = (u_1,\ldots,u_r)\in\FF_2^r$ with $\wt(u)\leq t$, and define an
{\em evaluation ideal}
\[
E_u = \langle x_1 - u_1,\ldots,x_r - u_r \rangle
\]
By definition of Evaluation Set, the partial assignment encoded by $E_u$ is consistent
with the syndrome constraints encoded by $L_s$. Equivalently,
\[
V_{\FF_2}(E_u + L_s) \neq \emptyset
\]
or, in ideal-theoretic terms, $1\notin E_u + L_s$.

If $H_{s,t}$ cannot be computed in practice, we can consider the identity
\[
H_{s,t} = \bigcup_{\wt(u)\leq t} V_{\FF_2}(J_{s,t} + E_u)
\]
Hence, instead of solving a single polynomial system in $n$ variables, one may solve a family
of systems, indexed by the vectors $u\in\FF_2^r$ with $\wt(u)\leq t$, each involving only
$n-r$ variables. We refer to this approach as a {\em hybrid strategy} for the Exact Syndrome
Decoding Problem.

Fix a vector $u = (u_1,\ldots,u_r)\in\FF_2^r$ such that $\wt(u)\leq t$.
We now study how assigning the variables $x_1,\ldots,x_r$ of the Evaluation Set
to the values $u_1,\ldots,u_r$, and thereby eliminating them from the system, affects
the computation of $V_{\FF_2}(J_{s,t} + E_u)$.

Consider the injective polynomial map $\psi_u:\FF_2^{n-r} \to \FF_2^n$ such that
\[
(v_1,\ldots,v_{n-r})\mapsto (u_1,\ldots,u_r,v_1,\ldots,v_{n-r})
\]
Put $t' = t - \wt(u)$ and denote by $H'_{t'}\subset \FF_2^{n-r}$ the Hamming variety
of weight $t'$. It is immediate that
\[
\psi_u(H'_{t'}) = H_t \cap \psi_u(\FF_2^{n-r})
\]
where $\psi_u(\FF_2^{n-r}) = V_{\FF_2}(E_u)$.

The algebraic counterpart of this result is obtained as follows.

\begin{proposition}
Denote $R' = \FF_2[x_{r+1},\ldots,x_n]$ and let $\cF'$ be the field equations
ideal of the subalgebra $R'\subset R$. Consider the surjective algebra homomorphism
$\varphi_u:R\to R'$ such that
\[
x_i\mapsto
\left\{
\begin{array}{cl}
u_i & \mbox{if}\ i\leq r, \\
x_i & \mbox{otherwise} \\
\end{array}
\right.
\]
Note that $E_u = \Ker\varphi_u$ and $\psi_u$ is the polynomial map corresponding to
the algebra homomorphism $\varphi_u$, namely
\[
\varphi_u(f)(v) = f(\psi_u(v))
\]
for all $f\in R$ and $v\in\FF_2^{n-r}$. We have that
\[
I'_{t'} + \cF' = \varphi_u(I_t + E_u + \cF)
\]
where $I'_{t'} + \cF'\subset R'$ is the Hamming ideal of weight $t'$.
\end{proposition}

\begin{proof}
It is sufficient to consider the variety identity $\psi_u(H'_{t'}) = H_t \cap \psi_u(\FF_2^{n-r})$,
together with the fact that $I'_{t'} + \cF', I_t + \cF$ and $E_u + \cF$ are radical ideals
corresponding to varieties $H'_{t'}, H_t$ and $\psi_u(\FF_2^{n-r})$, respectively.
\end{proof}

By defining $L'_{s,u} = \varphi_u(L_s)$ and
\[
J'_{s,t',u} = L'_{s,u} + I'_{t'}, H'_{s,t',u} = V_{\FF_2}(J'_{s,t',u})
\]
we obtain
\[
J'_{s,t',u} + \cF' = \varphi_u(J_{s,t} + E_u + \cF)
\]
or equivalently
\[
\psi_u(H'_{s,t'}) = H_{s,t} \cap \psi_u(\FF_2^{n-r})
\]
where $\psi_u(\FF_2^{n-r}) = V_{\FF_2}(E_u)$.

This identity reduces the Exact Syndrome Decoding Problem, namely the computation of the
syndrome variety $H_{s,t}$, to the computation of the varieties $H'_{s,t',u}$ corresponding
to each choice of $u \in \FF_2^r$ with $\wt(u)\leq t$ and $t' = t - \wt(w)$.
We refer to $J'_{s,t',u}$ and $H'_{s,t',u}$ as the {\em reduced syndrome decoding ideal
and reduced syndrome decoding variety}, respectively.

\section{An ISD-like strategy}

A more precise notation for the evaluation ideal should indicate its dependence
on the choice of an Evaluation Set $S = \{x_{i_1},\ldots,x_{i_r}\}\subset
X = \{x_1,\ldots,x_n\}$. Accordingly, we define
\[
E_{u,S} = \langle x_{i_1} - u_1,\ldots,x_{i_r} - u_r\rangle.
\]
Similarly, we write $J'_{s,t',u,S}$ and $H'_{s,t',u,S}$ to make their dependence
on $S$ explicit.

We have shown that to compute $H_{s,t}$ is equivalent to compute $H'_{s,t',u,S}$
for a fixed Evaluation Set $S$ and for all vectors $u\in\FF_2^r, \wt(u)\leq t$.
We have called this approach an {\em hybrid strategy}.

Since, in the Syndrome Decoding Problem, the weight $t$ denotes the error-correcting
capability of the code $\cC$, this parameter is typically small compared to the code
length $n$. Therefore, it is generally more efficient to fix a weight $\bt \leq \min(r,t)$
and check whether $H'_{s,t',u,S}\neq \emptyset$ for some vector $u\in\FF_2^r$ with
$\wt(u) = \bt$. If this is not the case, the Evaluation Set $S$ is updated. We call
this approach an {\em ISD-like strategy}, inspired by Information Set Decoding methods.

We remark that, in ISD methods, an Evaluation Set typically contains an Information Set
($r\geq k$), thereby reducing the decoding step to linear algebra. In our approach,
instead, an Evaluation Set is contained in an Information Set ($r\leq k$), since
we allow the solution of non-linear polynomial systems, for instance via \Gr\ basis
techniques.

An ISD-like strategy is described by the following algorithm. We assume that the parameters
$s,t$ are such that $\# H_{s,t} = 1$.

\begin{algorithm}[H]
\caption{\GBDecode}
\begin{algorithmic}[1]
\Require Ideal $J_{s,t}\subset R = \FF_2[X]$, parameters $r\leq k$ and $\bt \leq \min(r,t)$
\Ensure  \hspace{1pt} Unique vector $v\in\FF_2^n$ such that $H_{s,t} = \{v\}$
\State Set $t' = t - \bt$
\While{true}
  \State Sample uniformly an Evaluation Set $S = \{x_{i_1},\ldots,x_{i_r}\}\subset X$
  \State Let $S^c = \{x_{j_1},\ldots,x_{j_{n-r}}\}$ be the complement of $S$
  \For{\textbf{each} $u \in \FF_2^r$ with $\wt(u) = \bt$}
    \State Construct a generating set for the ideal $J'_{s,t',u,S}\subset R' = \FF_2[S^c]$ 
    \State Compute a \Gr\ basis $G$ of the ideal $J'_{s,t',u,S} + \cF'$
    \If{$G = \{x_{j_1} - v'_1,\ldots,x_{j_{n-r}} - v'_{n-r}\}$}
    \State Reconstruct $v = (v_1,\ldots,v_n)$ by setting
    \State \hspace{1em} $v_{i_\alpha} = u_\alpha$ ($1\leq\alpha\leq r$),
                        $v_{j_\beta} = v'_\beta$ ($1\leq\beta\leq n-r$)
    \State \Return $v$
    \EndIf
  \EndFor
\EndWhile
\end{algorithmic}
\end{algorithm}

In the above algorithm, the Evaluation Set $S$ is obtained by first sampling
an Information Set $\{x_{i_1},\ldots,x_{i_k}\}$, and then setting $S = \{x_{i_1},\ldots,x_{i_r}\}$
($r \leq k$). This is done by extracting the free variables of the linear system
defining the code $\cC$, namely
\[
l_1 = 0, \ldots, l_{n-k} = 0
\]
after performing Gaussian elimination on a randomly permuted ordering of the variable set
$X = \{x_1,\ldots,x_n\}$.

We will show that, under a unit-cost model for \Gr\ basis computations, the ISD-like strategy
with $\bt = 0$ requires no more such computations than the hybrid strategy.
Recall that the latter fixes the Evaluation Set $S$ and then performs an exhaustive search over
all vectors $u\in\FF_2^r$ satisfying $\wt(u)\leq t$.

\begin{proposition}
A single random choice of the Evaluation Set $S$ in algorithm \GBDecode,
leads to successful recovery of $v\in H_{s,t}$ with probability
\[
P(\bt) = \frac{\binom{t}{\bt}\binom{n-t}{r-\bt}}{\binom{n}{r}}
\]
\end{proposition}

\begin{proof}
Denote $S = \{x_{i_1},\ldots,x_{i_r}\}$ and $u = (v_{i_1},\ldots,v_{i_r})$.
A single iteration of the algorithm succeeds when the uniformly random set $S$
is such that $\wt(u) = \bt$.

Since $\wt(v) = t$, there are exactly $t$ coordinates of $v$ equal to $1$
and $n - t$ coordinates equal to $0$. Therefore, a successful Evaluation Set $S$
is obtained by choosing $\bt$ variables among the $t$ positions where $v$ has value 1
and $r - \bt$ variables among the $n - t$ remaining positions.
The number of successful choices is thus
\[
\binom{t}{\bt}\binom{n - t}{r - \bt}
\]
Since the total number of subsets of $\{x_1,\ldots,x_n\}$ of cardinality $r$ is
$\binom{n}{r}$, the claimed probability follows.
\end{proof}


Given the success probability $P(\bt)$, the expected number of \Gr\ basis
computations performed by the \GBDecode\ algorithm, that is, the average cost
of this algorithm under the assumption that each \Gr\ basis has unit cost,
is given by
\[
C(\bt) \approx \binom{r}{\bt} \cdot \frac{1}{P(\bt)} =
\frac{\binom{n}{r} \binom{r}{\bt}}{\binom{t}{\bt} \binom{n-t}{r-\bt}}
\]

\begin{proposition}
\label{isd-like-optimal}
The expected cost $C(\bt)$ of algorithm \GBDecode\ is minimized at $\bt = 0$.
\end{proposition}

\begin{proof}
It suffices to minimize
\[
F(\bt) = \frac{\binom{r}{\bt}}{\binom{t}{\bt}\binom{n-t}{r-\bt}}
\]
since $\binom{n}{r}$ does not depend on $\bt$. For integer arguments, the ratio
of consecutive terms is given by
\[
\begin{aligned}
R(\bt) & = \frac{F(\bt+1)}{F(\bt)}
= \frac{\binom{r}{\bt+1}}{\binom{r}{\bt}}
   \cdot \frac{\binom{t}{\bt}}{\binom{t}{\bt+1}}
   \cdot \frac{\binom{n-t}{r-\bt}}{\binom{n-t}{r-\bt-1}} \\
& = \frac{r-\bt}{\bt+1}
   \cdot \frac{\bt+1}{t-\bt}
   \cdot \frac{n-t-r+\bt+1}{r-\bt}
 = \frac{n-t-r+\bt+1}{t-\bt}
\end{aligned}
\]
We have $F(\bt+1) > F(\bt)$ if and only if $R(\bt) > 1$, that is
\[
n - t - r + \bt + 1 > t - \bt
\ \Longleftrightarrow\
\bt > t - \frac{n-r+1}{2}
\]

Since $t$ is, by hypothesis, the error-correction capability guaranteeing uniqueness
of the solution to the Syndrome Decoding Problem, the code $\cC$ has minimum distance
$d\geq 2 t + 1$. By the Singleton bound, we have $d\leq n - k + 1$, and hence
$2 t + 1 \leq n - k + 1$, that is
\[
t \leq \frac{n-k}{2}
\]
Since $r\leq k$, we obtain
\[
t - \frac{n-r+1}{2} \leq t - \frac{n-k+1}{2}
\]
Moreover, using $t \leq \frac{n-k}{2}$ we get
\[
t - \frac{n-k+1}{2}\leq \frac{n-k}{2} - \frac{n-k+1}{2} = -\frac12 < 0
\]
Thus the inequality $\bt > t - \frac{n-r+1}{2}$ holds for every $\bt\geq 0$, so
$F(\bt+1) > F(\bt)$ for all $0\leq \bt < \min(r,t)$. Hence $F(\bt)$ is strictly increasing
and its minimum is attained at $\bt = 0$.
\end{proof}

We finally compare the complexity of the ISD-like strategy of the \GBDecode\ algorithm,
run with the optimal parameter $\bt = 0$ established in Proposition \ref{isd-like-optimal},
with the complexity of the hybrid strategy, given by
\[
C' = \sum_{0\leq \bt\leq \min(r,t)} \binom{r}{\bt}
\]

\begin{proposition}
It holds that
\[
C(0) = \frac{\binom{n}{r}}{\binom{n-t}{r}} \leq C'
\]
with strict inequality whenever $r,t\geq 1$.
\end{proposition}

\begin{proof}
Since
\[
\frac{\binom{n}{r}}{\binom{n-t}{r}} = \frac{n! (n-t-r)!}{(n-r)! (n-t)!} = \frac{\binom{n}{t}}{\binom{n-r}{t}}
\]
we may rewrite
\[
C(0) = \frac{\binom{n}{t}}{\binom{n-r}{t}}
\]
Consider now the Vandermonde's identity,
\[
\binom{n}{t} = \sum_{0\leq \bt\leq t} \binom{r}{\bt} \binom{n - r}{t - \bt}
\]
where $\binom{r}{\bt} = 0$ for $\bt > r$, so the sum effectively ranges over
$0\leq \bt\leq \min(r,t)$.

As in the proof of Proposition \ref{isd-like-optimal}, the code $\cC$ has minimum
distance $d\geq 2 t + 1$. Moreover, by the Singleton bound, we have $d\leq n - k + 1$ and
therefore
\[
t\leq \frac{n-k}{2}\leq \frac{n-r}{2}
\]
since $r\leq k$. It follows that $t$ lies in the increasing region of the unimodal sequence
$i\mapsto\binom{n-r}{i}$. Consequently, for every $0\leq \bt \leq t$, we have
$t - \bt\leq t\leq \frac{n-r}{2}$, and hence
\[
\binom{n-r}{t - \bt}\leq \binom{n-r}{t}
\]

Substituting this estimate into Vandermonde's identity yields
\[
\binom{n}{t} = \sum_{0\leq \bt\leq \min(r,t)} \binom{r}{\bt}\binom{n - r}{t - \bt} \leq
\binom{n-r}{t} \sum_{0\leq \bt\leq \min(r,t)} \binom{r}{\bt} = \binom{n-r}{t} C'
\]
Hence, dividing by $\binom{n-r}{t} > 0$, we obtain $C(0)\leq C'$.

Finally, assume $r,t\geq 1$. Since $t - 1 < t\leq (n-r)/2$, we have
\[
\binom{n-r}{t-1} < \binom{n-r}{t}
\]
while $\binom{r}{1} = r > 0$. Therefore the inequality in the Vandermonde sum is strict,
which implies $C(0) < C'$.
\end{proof}

Note that the classical Information Set Decoding methods (see, for instance,
\cite{EW}) essentially correspond to choosing $r = k$. In particular, also setting
$\bt = 0$ in algorithm \GBDecode\ reduces to the Prange algorithm \cite{Pra}, which
pioneered the ISD paradigm.

The above results suggest that a Prange-like strategy is a natural candidate
for the optimization of the \GBDecode\ algorithm. This conclusion, however,
relies on the simplifying assumption that all \Gr\ basis computations
have the same cost.

In practice, one should investigate how this cost varies with
the choice of the parameter $r\leq k$ and the residual weight
$t' = t - \bt$. A more refined analysis must therefore account for the
actual cost of the individual \Gr\ basis computations.

In the following sections, we start addressing this issue by analyzing the
average number of \Gr\ basis calls required by the \GBDecode\ algorithm when
implemented using the \MultiSolve\ procedure \cite{LSPTV,LST}. In fact, the number
of such calls provides a first approximation of the computational complexity
of the algorithm. We then experimentally investigate the effect of tuning
\GBDecode\ for a random binary linear code with the NIST Security
Category 1 parameter set of the Classic McEliece cryptosystem.

\section{MultiSolve Algorithm}

In this section, we briefly review the \MultiSolve\ algorithm, a general method
for solving polynomial systems having at most one solution over a finite field.
The algorithm combines \Gr\ basis computations up to a prescribed solving degree
with a recursive branching strategy known as the ``multistep strategy''. Although
the algorithm is defined over arbitrary finite fields, we present it here only
over the binary field $\FF_2$, consistently with the framework adopted throughout
this paper. For a complete description of the algorithm, we refer the reader
to \cite{LST}.

Let $R = \FF_2[x_1,\ldots,x_n]$ and let $\cF = \langle x_1^2 - x_1,\ldots,x_n^2 - x_n \rangle$
be the field equations ideal. Consider $J = \langle f_1,\ldots,f_m \rangle\subset R$ an ideal
such that $\# V_{\FF_2}(J)\leq 1$.

Since $J + \cF$ is a radical ideal, a \Gr\ basis $G$ of the ideal $J + \cF$, with respect
to any monomial order of $R$, is either $G = \{1\}$ if $V_{\FF_2}(J) = \emptyset$, or
\[
G = \{x_1 - v_1,\ldots,x_n - v_n\}
\]
where $v = (v_1,\ldots,v_n)\in\FF_2^n$ such that $V_{\FF_2}(J) = \{v\}$.
Hence, in both cases, all polynomials in the \Gr\ basis have degree at most one.

Let $0\leq \ell \leq n$ and let $u = (u_1,\ldots,u_\ell)\in\FF_2^\ell$. We define
the evaluation ideal
\[
E_u = \langle x_1 - u_1, \ldots, x_\ell - u_\ell \rangle
\]
and the corresponding ideal $J_u = J + E_u$. By convention, when $\ell = 0$,
we set $E_{()} = 0$ and $J_{()} = J$. If $V_{\FF_2}(J) = \{v\}$, then
\[
V_{\FF_2}(J_u) =
\left\{
\begin{array}{cl}
V_{\FF_2}(J) & \mbox{if}\ u = (v_1,\ldots,v_\ell), \\
\emptyset & \mbox{otherwise}
\end{array}
\right.
\]
Therefore, we have
\[
V_{\FF_2}(J_u) = V_{\FF_2}(J_{(u,0)}) \cup V_{\FF_2}(J_{(u,1)})
\]
which motivates a binary divide-and-conquer strategy along the tree of partial assignments.

Given a generating set $H$ of an ideal and an integer $d\geq0$, we denote by
$\GB(H,d)$ the truncated \Gr\ basis computation with degree bound $d$, namely
the computation obtained by restricting the Macaulay matrices to degree
at most $d$. For a sufficiently large value of $d$, this procedure computes a complete
\Gr\ basis. The smallest such value is called a {\em solving degree} of $H$.

A first subroutine of algorithm \MultiSolve\ is the \GBSafe\ procedure, which performs
a truncated \Gr\ basis computation up to a prescribed degree $d$, limited by a timeout $\tau$.
It returns the flag $\tame$ if a complete \Gr\ basis is obtained, and $\wild$ otherwise.

\begin{algorithm}[H]
\caption{\GBSafe}
\begin{algorithmic}[1]
\Require Generating set $H$ of $J + \cF$; integer $d$; timeout $\tau > 0$
\Ensure  \hspace{1pt} $(\wild,\emptyset)$ or $(\tame,G)$ with $G = \GB(H,d)$
and $\maxdeg(G)\leq 1$
\State Compute $G\gets\GB(H,d)$ within timeout $\tau$
\If{the computation terminates within $\tau$ and $\maxdeg(G)\leq 1$}
  \State \Return $\tame,G$
\EndIf
\State \Return $\wild,\emptyset$
\end{algorithmic}
\end{algorithm}

The $\MultiSolve$ algorithm is a depth-first recursive procedure. Given a partial
assignment $u\in\FF_2^\ell$ of the first $\ell$ variables, it uses the predictive
function $\Oracle$ to decide whether to invoke $\GBSafe$. If a \Gr\ basis of $J_u + \cF$
cannot be computed, the algorithm recursively explores the two possible binary extensions
of the vector $u$. The objects $H$, $\Oracle$, $d$ and $\tau$ are treated as global
parameters.

\begin{algorithm}[H]
\caption{$\MultiSolve(u)$}
\label{alg:multisolve}
\begin{algorithmic}[1]
\Require A vector $u\in\FF_2^\ell$ ($0\leq \ell \leq n$)
\Ensure \hspace{1pt} A \Gr\ basis of $J_u + \cF$
\State $G\gets \{x_1 - u_1,\ldots,x_\ell - u_\ell\} \cup H$
\If{$\Oracle(G, \ell) = \tame$}
    \State $\status,G \gets \GBSafe(G,d,\tau)$
    \If{$\status = \tame$}
        \State \Return $G$
    \EndIf
\EndIf
\For{$b\in\FF_2$}
    \State $G' \gets \MultiSolve((u,b))$
    \If{$\maxdeg(G') = 1$}
        \State \Return $G'$
    \EndIf
\EndFor
\State \Return $\{1\}$
\end{algorithmic}
\end{algorithm}

Termination is guaranteed since, in the worst case, after $n$ recursive
steps all variables are assigned. In practice, however, the recursion is
typically terminated much earlier, as $\GBSafe$ may already compute a \Gr\
basis of $J_u + \cF$ at an intermediate node of the search tree.
The algorithm generalizes both exhaustive search over $\FF_2^n$, obtained
for an oracle that always returns \wild, and the classical hybrid strategy,
recovered by the oracle
\[
\OracleH_B(G,\ell) =
\left\{
\begin{array}{cl}
\wild & \mbox{if}\ \ell < B, \\
\tame & \mbox{otherwise}
\end{array}
\right.
\]

We briefly recall the complexity analysis of the \MultiSolve\ algorithm
given in \cite{LST}, in the simplified setting of a binary finite field.

The recursive calls of \MultiSolve\ form a full binary tree $T$, where
each internal node has exactly two children corresponding to the two recursive
calls generated by a $\wild$ case. This tree represents the execution in the
case where no solution exists, since all branches must be explored. If a
solution exists, the algorithm stops as soon as a successful branch is found,
and only a portion of the tree is visited.

\begin{proposition}
Let $N$, $M$, and $L$ denote the total number of nodes, internal nodes,
and leaves of a full binary tree $T$, respectively. The following relations
hold
\[
N = 2L - 1,\ M = L - 1
\]
\end{proposition}

The above is a well-known result. For a proof in the general (non-binary)
case, we refer the reader to \cite{LST}.

\begin{definition}
An oracle function \Oracle\ is {\em accurate} if it correctly predicts
the $\wild$ status of a \GBSafe\ call. It is {\em perfect} if it
correctly predicts the outcome of every \GBSafe\ call, whether $\tame$ or
$\wild$. Such a perfect oracle minimizes the total number of \GBSafe\ executions
among all possible oracle functions.
\end{definition}

We also define $\OracleT$ as the oracle function that always returns
the $\tame$ status. It represents the worst-case behavior of \MultiSolve,
as \GBSafe\ is invoked at every node of the recursion tree. On the other
hand, a perfect oracle represents the best-case scenario, where only the
necessary calls to \GBSafe\ are performed.
For simplicity, in the following analysis, we assume that the polynomial systems
under consideration have no solutions. Hence, \MultiSolve\ explores the full
recursion tree.

\begin{proposition}
Let $N$ be the number of \GBSafe\ executions under $\OracleT$ and let $L$
be the number of executions under a perfect oracle. Then
\[
\frac{N}{L} = 2 - \frac{1}{L}
\]
In particular, the maximum possible speedup is strictly smaller than
a factor of 2.
\end{proposition}

\begin{proof}
Under $\OracleT$, every node of the full binary recursion tree $T$ requires
a \GBSafe\ execution, while a perfect oracle avoids all internal nodes and
requires only the computations associated with the leaves. Hence, $N$ and
$L$ are respectively the number of nodes and leaves of $T$. Since $T$ is
a full binary tree, we have $N = 2L - 1$, and the result follows.
\end{proof}

This result extends to arbitrary finite fields \cite{LST}.
The preceding analysis shows that the cost of \MultiSolve\ is essentially
dominated by the number of tame nodes in the recursion tree, since each of
them corresponds to a complete \Gr\ basis computation performed by \GBSafe.
Since all calls to \GBSafe\ are subject to the same fixed timeout,
we adopt a simplified cost model in which each execution of \GBSafe\ has
unit cost.

In the next section, we apply \MultiSolve\ within the \GBDecode\ algorithm
for computing \Gr\ bases. For different choices of parameters, such as $r\leq k$,
the combinatorial cost $C(\bt)$ of the algorithm is therefore adjusted by
a multiplicative factor given by the average number of tame cases encountered
during the execution of \MultiSolve. Note that, when applying \MultiSolve\ to
\GBDecode, the number of variable assignments with $b = 1$ must be bounded by the residual
weight $t' = t - \bt$.


\section{Experimental results}

We have implemented the C-Hamming, FC-Hamming, and QFC-Hamming ideals
introduced in Sections 3, 4, and 5 in C++, together with the \MultiSolve\ and \GBDecode\ algorithms.
For \GBSafe, we use our own C++ implementation of F4 \cite{Fau99}, with M4RI \cite{M4RI}
for dense linear algebra over $\FF_2$. We have optimized its application to each
of the three Hamming ideals.
We investigate the performance of this algebraic approach when the solving step of \GBDecode\
is carried out by \MultiSolve.

All experiments use the NIST Security Category 1 parameter set of the Classic McEliece
cryptosystem \cite{ClassME}, namely
\[
n = 3488, k = 2720, t = 64
\]
for a binary linear code drawn uniformly at random with these parameters. No structure
of the code is exploited by \GBDecode. The expected number of additional vectors of weight
$t$ having the same syndrome is given by (see, for instance, \cite{GLMS})
\[
\frac{\binom{n}{t} - 1}{2^{n-k}} \approx 2^{-312}
\]
Thus, with overwhelming probability, there is no other vector of weight $t$ with the
same syndrome, so the solution is unique, as assumed in Section 7.

\subsection*{Setup}

One run of the implementation corresponds to one iteration of the \textbf{while} loop
of \GBDecode. An Evaluation Set $S = \{x_{i_1},\ldots,x_{i_r}\}\subset X$ is obtained from
the free variables resulting from Gaussian elimination applied to a randomly permuted
ordering of $X$, and one vector $u\in\FF_2^r$ with $\wt(u) = \bt$ is examined. For $\bt = 0$,
this vector is $u = 0$. For $\bt > 0$, it is chosen uniformly at random among the
$\binom{r}{\bt}$ vectors of weight $\bt$. The algorithm \GBDecode\ examines all such vectors,
resulting in the combinatorial cost $C(\bt)$ defined in Section 7. A single integer seed
determines the binary linear code, the error vector, the Evaluation Set, and the vector $u$,
making each run exactly reproducible. The same seeds are used for the three Hamming ideals,
allowing them to be compared instance by instance. At $r = k - 10 = 2710$ and $\bt = 0$,
we use $50$ instances for C-Hamming and $200$ for each of FC-Hamming and QFC-Hamming.

The parameters of \GBSafe\ are a timeout $\tau = 180$ seconds and a degree bound $d = 10$.
Note that for the FC-Hamming ideal, the maximal degree of the generators is $10$.
The oracle used is $\OracleT$, so that the nodes of the recursion tree of \MultiSolve\
are exactly the calls to \GBSafe: the \wild\ calls are its internal nodes and the tame
calls its leaves, as in Section 8. A call is \wild\ if the prescribed computational limits are reached
without obtaining a linear \Gr\ basis. The \MultiSolve\ recursion progressively
assigns values to at most $k - r$ variables in the Information Set that do not belong to $S$.

All results are obtained using degree reverse lexicographic order with the native
variable ordering supplied by the models.
The computations were run on an Intel Core i9-10900 at $2.80$ GHz, with
$10$ physical cores and $128$ GB of memory.
An additional $250$ tests with \Magma\ \cite{MAGMA} confirmed the comparative assessment
of the three models, with FC-Hamming and QFC-Hamming performing best.

\subsection*{One iteration of \GBDecode\ at $r = k - 10$}

For $r = k - 10 = 2710$, with $n - r = 778$ remaining variables and residual weight
$t' = 64$ (Prange case, $\bt = 0$), the three Hamming ideals yield the systems
reported in Table \ref{tab:main}. They are constructed using a recursive binary decomposition
of the remaining variables into halves, with
\[
L = \lceil \log_2(\max(t', n-r-t') + 1) \rceil - 1 = 9
\]
as prescribed in Section 2.

In every instance and for each of the three ideals, the call to \GBSafe\ at the root
returned $\wild$, its \Gr\ basis not being completed within the prescribed limits. This is
precisely the situation \MultiSolve\ is designed for: assigning few variables replaces
the original system with a family of simpler ones. The number of tame calls determines
the cost of one iteration, as in Section 8.

Table \ref{tab:main} reports the average number of tame calls, together with the mean depth
of the tame calls and the solving degree actually observed in the corresponding \Gr\ basis
computations performed by our implementation. For the FC-Hamming and QFC-Hamming ideals, the reported
values are averaged over $200$ instances. Computations with the C-Hamming ideal were significantly
more time-consuming, and we therefore restricted the experiments to $50$ instances. The FC-Hamming
ideal required the fewest tame calls on average. A tame \GBSafe\ run with the FC-Hamming ideal takes
about $100$ seconds on average.

\begin{table}[H]
\begin{center}
\small
\begin{tabular}{lcccccc}
\hline
& \multicolumn{3}{c}{system at the root}
& \multicolumn{3}{c}{recursion tree} \\
Hamming ideal & vars & gens & $\maxdeg$ & tame calls & tame depth & sol.\ deg. \\
\hline
C-Hamming   & $7544$ & $6776$ & $2$  & $25.20$ & $4.78$ & $3$ \\
FC-Hamming  & $2844$ & $2076$ & $10$ & $2.52$  & $1.41$ & $10$ \\
QFC-Hamming & $4133$ & $3365$ & $2$  & $4.31$  & $2.14$ & $3$ \\
\hline
\end{tabular}
\end{center}
\caption{The three Hamming ideals for $r = 2710$ and $\bt = 0$. The table reports the systems at the
root of the recursion tree (excluding the $n-k$ linear equations of the code and the field
equations) and the corresponding tame statistics of \MultiSolve\ for a single iteration
of \GBDecode.}
\label{tab:main}
\end{table}

The variable counts in Table \ref{tab:main} include the $X$-variables and only the
auxiliary variables needed to express the Hamming constraints.

\begin{table}[H]
\begin{center}
\small
\begin{tabular}{lrrrrrrrr}
\hline
& \multicolumn{8}{c}{depth} \\
& $0$ & $1$ & $2$ & $3$ & $4$ & $5$ & $6$ & $7$ \\
\hline
C-Hamming   & $0$ & $0$ & $0$ & $0.08$ & $7.58$ & $15.54$ & $1.92$ & $0.08$ \\
FC-Hamming  & $0$ & $1.48$ & $1.04$ & --- & --- & --- & --- & --- \\
QFC-Hamming & $0$ & $0$ & $3.69$ & $0.62$ & --- & --- & --- & --- \\
\hline
\end{tabular}
\end{center}
\caption{Average number of tame calls at each depth of the recursion tree,
at $r = 2710$ and $\bt = 0$, over all runs. A dash means that no tree reaches that depth.}
\label{tab:depth}
\end{table}

\subsection*{Recovering the Prange algorithm}

For $r = k$ and $\bt = 0$, the algorithm \GBDecode\ reduces to the Prange algorithm
\cite{Pra}, as observed in Section 7. The Evaluation Set is then an entire Information Set,
and the code equations determine all remaining $X$-variables uniquely. The Hamming
condition can therefore be checked by computing the weight of the solution vector. Thus FC-Hamming and QFC-Hamming lead to
the same computation, and the recursion tree consists of a single node.
Using M4RI \cite{M4RI}, solving the linear system and checking the Hamming
weight takes about $3$ milliseconds on average.
This cost should be multiplied by the combinatorial cost $C(0)$, which for the
parameters considered here is
\[
C(0) = 2^{142.78}
\]

\subsection*{Choosing $r$ and $\bt$}

Table \ref{tab:params} reports, for $r = k = 2720$ and $r = k - 10 = 2710$, the average number
of tame calls with the FC-Hamming ideal, together with the expected combinatorial cost
$C(\bt)$ of Section 7. Lowering $r$ from $k$ to $k-10$ decreases $\log_2 C(0)$ by $1.25$ bits,
but the average number of tame calls increases from $1$ to $2.52$. Thus, the additional \Gr\ basis
computations reduce the overall benefit of the combinatorial gain. Taking $\bt = 2,20,30$,
at $r = 2710$, gives respectively $2.52$, $2.24$, and $2.34$ tame calls on average, while
$\log_2 C(\bt)$ increases by about $7$, $75$, and $117$ bits relative to $\bt = 0$. For the parameters
considered here, these experiments therefore indicate that, among the choices tested,
$r = k$ and $\bt = 0$ is the most efficient one.

\begin{table}[H]
\begin{center}
\small
\begin{tabular}{lrrrrr}
\hline
$r$                & $2720$   & $2710$   & $2710$   & $2710$   & $2710$   \\
$\bt$              & $0$      & $0$      & $2$      & $20$     & $30$     \\
\hline
FC-Hamming         & $1$      & $2.52$   & $2.52$   & $2.24$   & $2.34$   \\
$\log_2 C(\bt)$    & $142.78$ & $141.54$ & $148.53$ & $216.35$ & $258.66$ \\
\hline
\end{tabular}
\end{center}
\caption{Average number of tame calls with the FC-Hamming ideal, together with $\log_2 C(\bt)$,
where $C(\bt)$ is the expected combinatorial cost, for different choices of $r$ and $\bt$.}
\label{tab:params}
\end{table}

\subsection*{Stability of the experiments}

To assess the stability of the observed averages with respect to the sample size, we ran the experiments on $200$
instances and recomputed the averages after $50$, $100$, $150$, and $200$ instances. Table \ref{tab:stab} reports
the average number of tame calls for $r = 2710$ and $\bt = 0$. The values remain fairly stable as further instances
are added. The mean numbers of tame calls in the two batches of $100$ instances are $2.49$ and $2.55$
for FC-Hamming, and $4.27$ and $4.35$ for QFC-Hamming, indicating consistency between batches.

Each average is reported together with an approximate $95\%$ confidence interval, computed as the sample mean plus
or minus $1.96$ times its standard error. The standard error is the sample standard deviation divided by the square
root of the number of instances. Over the full set of $200$ instances, the sample standard deviations are $0.63$
for FC-Hamming and $0.52$ for QFC-Hamming, yielding approximate $95\%$ confidence intervals
with half-widths of $0.09$ and $0.07$ tame calls, respectively.

\begin{table}[H]
\begin{center}
\small
\begin{tabular}{lcccc}
\hline
instances   & $50$ & $100$ & $150$ & $200$ \\
\hline
FC-Hamming  & $2.44 \pm 0.16$ & $2.49 \pm 0.12$ & $2.53 \pm 0.10$ & $2.52 \pm 0.09$ \\
QFC-Hamming & $4.20 \pm 0.11$ & $4.27 \pm 0.09$ & $4.26 \pm 0.08$ & $4.31 \pm 0.07$ \\
\hline
\end{tabular}
\end{center}
\caption{Average number of tame calls at $r = 2710$ and $\bt = 0$, with approximate $95\%$ confidence intervals.}
\label{tab:stab}
\end{table}

These experiments show the two components of the cost of \GBDecode. An Evaluation Set strictly contained
in an Information Set gives a combinatorial gain of $1.25$ bits at $r = k - 10$, together with an
algebraic cost in the inner loop, measured by the tame calls of \MultiSolve. The balance between these
two components depends on the Hamming ideal and on the solver, as shown by the different behavior
of the three constructions. A Hamming ideal requiring fewer tame calls, or a faster solver, would reduce
the algebraic cost and could make smaller values of $r$ competitive.

\section{Conclusions and further directions}

We have presented an algebraic approach to the Syndrome Decoding Problem, introducing
various ideals defining the Hamming variety based on the theory of elementary symmetric
functions and factorizations derived from Lucas' theorem. These constructions are integrated
with the Information Set Decoding paradigm. The resulting \GBDecode\ algorithm introduces
a tunable trade-off between combinatorial search and algebraic solving, while \MultiSolve\
replaces a hard polynomial system with a family of simpler systems obtained by progressively
assigning variables.

Experiments on random binary linear codes with Classic McEliece Category~1 parameters
show that the choice of the Hamming representation has a significant impact on the
computational cost. Among the three constructions, FC-Hamming produced the smallest recursion trees on average. For the parameter choices considered in the experiments,
the Prange-like configuration $r = k$ and $\bt = 0$ yielded the lowest overall cost.

These results support the viability of the proposed framework and suggest several directions
for further investigation, including improved Hamming representations, more efficient
\Gr\ basis solving strategies, and a broader exploration of the trade-off between the
combinatorial and algebraic components of \GBDecode.

\section*{Acknowledgements}

The authors would like to thank Lorenzo D'Ambrosio and Gabriele Mancini for carefully
reading parts of the manuscript and for their helpful comments and suggestions. The first
author gratefully acknowledges Raffaele Vitolo for his hospitality and support during
a research stay at the University of Salento.

\end{document}